\documentclass[a4paper,11pt]{article}
\pdfoutput=1 

\usepackage{jheppub} 

\usepackage[T1]{fontenc} 
\usepackage{amsthm}

\title{\boldmath Normalized Fuchsian form on Riemann sphere and differential equations for multiloop integrals.}

\author{Roman N. Lee }
\author{and Andrei A. Pomeransky}

\affiliation{The Budker Institute of Nuclear Physics, 630090, Novosibirsk, Russia}

\emailAdd{r.n.lee@inp.nsk.su}
\emailAdd{a.a.pomeransky@inp.nsk.su}

\newtheorem*{prop}{Proposition}

\abstract{
	We consider the question of reducibility of the differential system to normalized Fuchsian form on the Riemann sphere. The differential equations for the multiloop integrals in $\epsilon$-form constitute a particular example of the  normalized Fuchsian form.  We formulate the algorithmic criterion of reducibility. We also consider the question of the proper choice of variable in the differential system suitable for its reduction to $\epsilon$-form.
}

\begin{document} 
\maketitle
\flushbottom
\section{Introduction}\label{sec:intro}

The reduction of differential systems to various normal forms is a classical problem of differential equations theory. It is very important for applications. In particular, recent progress with multiloop calculations is connected with the possibility to reduce the system for the master integrals to $\epsilon$-form \cite{Henn2013}. In Ref. \cite{Lee2014} an efficient algorithm of finding this $\epsilon$-form has been formulated. This algorithm is based on the transformations which are singular in exactly two points (such transformations are called balances in Ref. \cite{Lee2014}). It is expected to happen, and indeed happens in real-life examples, that this $\epsilon$-form can not be found. In particular, even the global Fuchsian form for the differential equations with regular singularities may be unreachable due to the result of Bolibrukh \cite{Bolibrukh1989}. The algorithm of Ref. \cite{Lee2014} lacks an essential ingredient: the termination criterion. In practice, the irreducible cases are easily detected by trying several different ways to make the reduction and always failing to construct a required balance at some step. Nevertheless, the decisive criterion of the (im)possibility to find the required $\epsilon$-form is of essential interest. The main goal of the present paper is to formulate this criterion.

In fact, the formulated criterion concerns a more general class of the normal forms, the global normalized Fuchsian forms, of which the $\epsilon$-form is a specific example. We will define these forms later.

\section{Preliminary considerations}

We consider the system of the form
\begin{equation}
\partial_x J = M J\,.\label{eq:DS1}
\end{equation}
Here $J$ is a column of the unknown functions and $M$ is a matrix with entries being the rational functions of $x$. We are interested in the reduction of the system by the rational transformations of functions $J= T\tilde{J}$, where $T$ is a matrix with entries being the rational functions of $x$. The new functions $\tilde{J}(x)$ obey the differential system
\begin{equation}
\partial_x \tilde{J} = M_T \tilde{J}\,,
\end{equation}
where 
\begin{equation}
	M_T = T^{-1}(MT-\partial_x T)\,.
\end{equation}

The system \eqref{eq:DS1} is said to be in Fuchsian form at $x=x_0\neq\infty$ iff the matrix $M$ has a first-order pole as a function of $x$, i.e., iff 
\begin{equation}
	M(x)=\frac{A}{x-x_0}+O\left((x-x_0)^0\right)\,.
\end{equation}
We will call $A$ in the above formula \emph{the matrix residue of $M$ at $x_0$}. By a slight abuse of notations we will also say that the matrix $M$ is Fuchsian at $x=x_0$. 

If the system is in Fuchsian form at $x=x_0$, the point $x_0$ is necessarily regular singular, i.e. the general solution is bounded by a power of $|x-x_0|$ when $x\to x_0$ (from within some angular sector). Vice versa, at any regular singular point $x=x_0$ the system can be \emph{reduced} to Fuchsian form by known algorithm \cite{barkatou2009moser}. In fact, if all \emph{finite} points are regular, the algorithm of Ref. \cite{barkatou2009moser} allows to reduce the system to Fuchsian form at all these points. In what follows we assume that the differential system \eqref{eq:DS1} has only regular singular points (possibly, including $x=\infty$).

Note that the spectrum of the matrix residue at a given point is not invariant under the rational transformations preserving the Fuchsian form at  this point. However, if we replace each eigenvalue $\lambda_i$ by the equivalence class $\Lambda_i=\lambda_i+\mathbb{Z}$, the set $\{\Lambda_1,\ldots \Lambda_n\}$ will be invariant, including the multiplicities of each class. By the \emph{normalization} we will call the rule which allows to pick exactly one specific representative out of each distinct equivalence class $\Lambda_i$. 

Let us give two examples of the normalizations. Let the entries of the matrix $M$ belong to the field of rational functions of $x$ over $\mathbb{C}$. Then we might choose the normalization by the following condition: the matrix residue is said to be normalized iff all its eigenvalues, $\lambda$, satisfy $0\leqslant \Re \lambda <1$. More generally, we might choose the normalization by the condition $a\leqslant \Re \lambda <a+1$, where $a$ is some real number.

Second example of normalization concerns the differential systems appearing in the multiloop calculations. The matrix $M$ depends now on the parameter $\epsilon$. Reduction to $\epsilon$-form may be possible only if all eigenvalues of the matrix residue at a given point have the form $k+\alpha \epsilon$, where $k\in \mathbb{Z}$, and $\alpha \in \mathbb{C}$. In this case we may fix the normalization by the requirement  that eigenvalues are proportional to $\epsilon$.

The notion of normalization is very broad, and in the following we will avoid the specification of the normalization rules, when possible. The only properties of the normalization which will be important for us are the following: if the set ${\lambda_1,\ldots \lambda_n}$ is normalized, then we might claim that:
\begin{enumerate}
	\item There are no resonances, i.e., from $\lambda_i-\lambda_j\in\mathbb{Z}$ it follows that $\lambda_i=\lambda_j$.
	\item The set ${\lambda_1+k_1,\ldots \lambda_n+k_n}$ for $k_1,\ldots k_n\in\mathbb{Z}$ is normalized only if $k_1=\ldots k_n=0$.
\end{enumerate}
%

We will say that the system \eqref{eq:DS1} in Fuchsian form at $x=x_0\neq \infty$ is normalized at $x=x_0$  iff the matrix residue $A$ in this point is normalized. Again, by a slight abuse of notations we will say that the matrix $M$ is normalized at $x=x_0$. Note that the normalization rules may, in principle, vary from point to point.

The property of being in (normalized) Fuchsian form with the matrix residue $A$ at the special point $x=\infty$ is defined as the corresponding property at $y=0$ of the system obtained by going over to new variable $y=1/x$.

Let us now prove the following
\begin{prop}
	If both $M$ and $M_T$ are normalized by the same rule at $x=x_0$, then $T$ is regular and invertible at $x=x_0$.
\end{prop}
\begin{proof}
	Let us prove by contradiction. Let $M$ and $M_T$ be normalized Fuchsian at $x=x_0$, but $T$ and/or $T^{-1}$ be singular at $x=x_0$.
	In what follows we consider the case when $x_0\neq\infty$ and  $T$ is singular at $x=x_0$. Other cases can be considered in a similar way. Let $M=\frac{A_{0}}{x-x_0}+O((x-x_0)^0)$, $M_{T}=\frac{B_{0}}{x-x_0}+O((x-x_0)^0)$ are  normalized at $x=x_0$, but $T$ is singular, $T\left(x\right)=\frac{T_{0}}{(x-x_0)^{n}}+\ldots,\quad\left(n>0\right)$.
	Leaving only the most singular in $x-x_0$ terms in the identity $MT=TM_T+\partial T$,
	we obtain
	\[
	A_{0}T_{0}=T_{0}(B_{0}-n)
	\]
	
	The matrix $B_{0}$ is, in general, reducible to Jordan form. It
	is easy to see that there exists a generalized eigenvector $u$ of $B_0$, such that
	\begin{equation}\label{eq:cond1}
	B_{0}u=\lambda u+u_1,
	\end{equation}
	and
	\begin{align}
	\label{eq:cond2}
 	 T_{0}u\neq0\,,\\
 	 \label{eq:cond3}
	T_{0}u_1=0\,.
	\end{align}
	Here $u_1$ is either zero or another generalized eigenvector, $\lambda$ is some eigenvalue of $B_0$ which, by the assumption, satisfies the normalization criterion.
	Indeed, the generalized eigenvectors of $B_0$ form a basis, therefore, there must be one for which \eqref{eq:cond1} and \eqref{eq:cond2} hold. If condition \eqref{eq:cond3} fails, we may try $u_1$ as a new candidate for $u$. Clearly, the process terminates and all three conditions \eqref{eq:cond1}-\eqref{eq:cond3} hold.
	Then we have
	\begin{equation}
	A_{0}T_{0}u=T_{0}(B_{0}-n)u=(\lambda-n)T_{0}u
	\end{equation}
	Therefore $A_{0}$ has the eigenvector $T_{0}u$ with the eigenvalue $\lambda-n$. But $\lambda$  and $\lambda-n$ can not both belong to norlmalized set since they differ by integer number $n\neq 0$. Therefore, $A_{0}$ is not normalized, which is the contradiction.
\end{proof}

Let us now make the following observation: \emph{Any differential system with regular singularities on the Riemann sphere can be reduced  to normalized Fuchsian form in all  singular points but maybe one  (called below the \emph{exceptional} point) by means of known algorithms}.
Indeed, the Fuchsian form in all points but one point can be achieved by means of Barkatou\&Pfl\"ugel algorithm \cite{barkatou2009moser}, or by the analogous algorithm presented in first subsection of Section 3 of Ref. \cite{Lee2014}. Then, the normalization of the matrix residues can be performed by balances  shifting the eigenvalues by $\pm 1$ as explained in second subsection of Section 3 of Ref. \cite{Lee2014}. Both for achieving Fuchsian form and normalization at a given point, we use the balances between this point and the exceptional one, without taking into account the behavior at the exceptional point. Here we intentionally do not refer to the global reduction described in Section 4 of Ref. \cite{Lee2014} to simplify the presentation. In actual calculations to facilitate the computation, one is free, and, likely, is forced to adjust the balances in the second point following the prescriptions of Section 4 of Ref. \cite{Lee2014}.

\section{Global normalized Fuchsian form}

According to the observation of the previous Section, given a system \eqref{eq:DS1} with regular singularities, we may perform the reduction to the normalized Fuchsian form in all points but one exceptional. Therefore, without loss of generality, we may assume that the system \eqref{eq:DS1} is in normalized Fuchsian form everywhere except $x=\infty$. Thanks to Proposition 1 of the previous Section, the transformation $T$ reducing $M$ to global normalized Fuchsian form, if it exists, is regular and invertible everywhere except $x=\infty$. I.e., both $T$ and $T^{-1}$ are matrices with entries being polynomials of $x$, which also implies that $\det {T}$ is independent of $x$. In order to find this transformation, we use the following approach. Let us first use the observation of the previous Section and the possibility to choose another exceptional point. Without loss of generality we may choose it as $x=0$. Therefore, we find the transformation $U$ such that 
\begin{equation}\label{eq:M_U}
	\tilde{M}=M_U=U^{-1}(MU -\partial_x U)
\end{equation} is normalized Fuchsian everywhere except $x=0$. Now we note that, if $T$ exists, the transformation $S=U^{-1}T$ is necessarily polynomial in $x^{-1}$ with constant determinant as it should be regular and invertible everywhere except $x=0$. Therefore, if $T$ exists, we may represent 
\begin{equation}\label{eq:decomposition}
U =T(x)S^{-1}\left(x^{-1}\right)
\end{equation}
where both $T$ and $S$ are polynomial matrices of their arguments ($x$ and $x^{-1}$, respectively) with the determinants independent of $x$. This decomposition explicitly demonstrates that $\det U$ should be independent of $x$ and that $U$ and $U^{-1}$ are necessarily regular in all points but two: $x=0$ and $x=\infty$, so, their entries are  Laurent polynomials in $x$. Let us note that if the decomposition exists, it is unique, up to the transformations $T(x)\to T(x) L\,,\ S(x^{-1})\to S(x^{-1}) L$, where $L$ is constant matrix. Indeed, suppose, there is another decomposition $U=T_1(x)S^{-1}_1(1/x)$. Then we have $T^{-1}T_1=S^{-1}S_1$. The left-hand side of this identity is polynomial in $x$, while the right-hand side is polynomial in $x^{-1}$, therefore they both are equal to some constant matrix $L$.

Note that it is very instructive to view the matrix $U$ in \eqref{eq:decomposition} as the transition function for the holomorphic vector bundle on the Riemann sphere covered by two charts obtained by stereographic projections from North and from South pole. This perspective is in no way accidental, in fact the problem of finding the global normalized Fuchsian form can be reformulated on the very early stage in the language of holomorphic vector bundles on the Riemann sphere (cf. \cite{Deligne}). According to the Birkhoff-Grothendieck theorem (see, e.g., Ref. \cite{Hitchin}), any holomorphic vector bundle on the Riemann sphere is equivalent to the direct sum of line bundles, each characterized by an integer number, its degree. This corresponds to the existence of the decomposition 
\begin{equation}\label{eq:decomposition1}
U =T(x)x^DS^{-1}(1/x)\,,
\end{equation}
here $T$ and $S$ are polynomial matrices of their arguments ($x$ and $x^{-1}$, respectively) with the determinants independent of $x$, and $D=\mathrm{diag}(d_1,d_2,\ldots)$, where $d_1,d_2,\ldots$ are the degrees of the line bundles entering the direct sum in  Birkhoff-Grothendieck theorem. So, the equivalence classes of the holomorphic vector bundles on the Riemann sphere are labeled by the spectrum of $D$. The case $D=0$, when Eq. \eqref{eq:decomposition1}, reduces to  Eq. \eqref{eq:decomposition} corresponds to the triviality of the vector bundle. The algorithm of finding the decomposition \eqref{eq:decomposition1} can be constructed along the lines of Ref. \cite{HAZEWINKEL1982207}. Let us present, for the sake of the completeness, the algorithm solving the following simpler problem\footnote{In fact, this is a variant of the Riemann-Hilbert problem.}: \emph{Given a rational $n\times n$ matrix $U(x)$ find the decomposition \eqref{eq:decomposition} or prove that it does not exist}.

\paragraph{Finding decomposition \eqref{eq:decomposition}} 
First we check that $U$ and $U^{-1}$ have entries which are Laurent polynomials in $x$. Note that this necessarily should be the case when $U$ comes from the reduction of the differential equations. Then we check that $\det U$ is independent of $x$. If it is not so, the decomposition does not exist. Then we try to find polynomial vectors $v(x)$, such that $U^{-1}(x)v(x)$ contains only negative powers of $x$. Let us denote by $\mathrm{PP}_\infty f(x)$ the principal part of the expansion of $f(x)$ at $x=\infty$, i.e., the contribution of the positive powers. Then we have the equation
\begin{equation}\label{eq:PPeq}
	\mathrm{PP}_\infty U^{-1}(x)v(x) =0
\end{equation}
If the decomposition exists, it is clear that \eqref{eq:PPeq} has $n$ linearly independent solutions. Indeed, the columns of the matrix $T$ constitute the required set. Now we show that the converse is also true. Let us define $T$ as the matrix with columns being the found $n$ vectors $v_1(x)\ldots v_n(x)$ Then we have $\mathrm{PP}_\infty U^{-1}(x)T(x) =0$, i.e.
\begin{equation}
 U^{-1}(x)T(x) =Q(x^{-1})
\end{equation}
where $Q(x^{-1})$ is polynomial in $x^{-1}$. Taking the determinant of both sides, we have  
\begin{equation}
(\det U)^{-1}\det T(x) =\det Q(x^{-1})\,.
\end{equation}
The left-hand side of this identity is polynomial in $x$ (since $\det U$ is independent of $x$), while the right-hand side is polynomial in $x^{-1}$, therefore they both are equal to some constant $A$. Since $T$ is constructed of $n$ linearly independent vectors, $A\neq0$. Therefore, we have proved that both $\det T(x)$ and $\det Q(x^{-1})$ are independent of $x$. Taking $S=Q$, we have the required decomposition. Note that the search of the vectors $v$ is reduced to a finite system of linear equation. This is due to the restriction on the  maximal power of the polynomials in $v(x)$: from $T=U S$ we see that the maximal degree of the polynomials in $v(x)$ is restricted by that in $U$ (recall that $S$ is polynomial in $x^{-1}$). 

We have implemented the described algorithm of finding the decomposition \eqref{eq:decomposition} in the Wolfram Mathematica notebook (see ancillary file for this submission).

\section{Application to $\epsilon$-form}

In application to the algorithm of Ref. \cite{Lee2014}, our present result allows one to perform the second step, normalization, or to prove that the global normalized Fuchsian form does not exist. The latter, in particular, means that the $\epsilon$-form does not exist. Note, that in order to achieve $\epsilon$-form, one has to perform third step, which is the factorization of $\epsilon$ using the transformation independent of $x$. Taking into account the proposition of Section 2, we can see that, indeed, the transformation to $\epsilon$-form, if it exists, should be independent of $x$. Therefore we can claim that if the algorithm of Ref. \cite{Lee2014} fails at the stage of factorization, $\epsilon$-form can not be reached by rational transformations.

\subsection*{Change of variable}

So far, our considerations related the reducibility to $\epsilon$-form by the transformations $T$ which are rational functions of the given variable. If $\epsilon$-form can not be obtained, one might ask whether it is possible to change variable in such a way that the transformations rational in new variable are sufficient to obtain $\epsilon$-form. Let us therefore consider the change of variable 
\begin{equation}\label{eq:varchange}
x=p(y)/q(y)\,,
\end{equation}
where $p(y)$ and $q(y)$ are two coprime polynomials in $y$. Without loss of generality we may assume that $\deg q\leqslant \deg p$.
Note that by passing to $y$ we extend the class of transformations as any transformation $T$ rational in $x$ is also rational in $y$, but the converse is not true in general.

Suppose first that $\epsilon$-form can not be reached by rational transformations in $x$ because not all eigenvalues of the matrix residues have the form $k+\alpha \epsilon$, with integer $k$. We will restrict ourselves to the case relevant for multiloop calculations when there are singular points in which the eigenvalues of the matrix residue can be represented as  $k/2+\alpha \epsilon$, where $k$ is odd at least for one eigenvalue. We will try to find a change of variable of the form \eqref{eq:varchange} such that after the change all residues have the desired form $k+\alpha \epsilon$ with integer $k$. Let us, e.g., assume, that this happens at $x=0$. Then any preimage of $x=0$ under the variable change \eqref{eq:varchange} should be a multiple zero of $p(y)$ with multiplicity equal to even number. In other words, $p(y)=\tilde{p}^2(y)$, where $\tilde{p}(y)$ is a polynomial. In general, if we have $n$ points $x=x_i\,,\ (i=1,\ldots,n)$ in which the matrix residues have half-integer eigenvalues (at $\epsilon=0$), we have the following constraints:
\begin{equation}\label{eq:cond}
\beta_i p(y) - \alpha_i q(y) = \tilde{p}_i^2(y)\,,\quad (i=1,\ldots, n)\,.
\end{equation}
Here $[\alpha_i:\beta_i]$ are the homogeneous coordinates of $x_i=\alpha_i/\beta_i$. The suitable variable changes depend on the number of points $n$. 
\begin{enumerate}
	\item $n=1$: we assume that $x_1=0$, then any suitable variable change has the form $x=p^2/q$	
	\item $n=2$: we assume that $x_1=0\,,\ x_2=\infty$, then any suitable variable change has the form $x=p^2/q^2$	
	\item $n=3$: we assume that $x_1=0\,,\ x_2=\infty\,,\ x_3=1$, then any suitable variable change has the form $x=\left(p^2+q^2 \over 2 p q\right)^2$.
	\item $n\geqslant4$: no suitable variable change (this is a simple consequence of Eq. \eqref{eq:cond} and Lemma 2.3 of Ref.\cite{Reid}).
\end{enumerate}
Here we used the fact that any triplet of points on the Riemann sphere can be mapped onto $\{0,1,\infty\}$ by M\"obius transformations, $p$ and $q$ are polynomials in $y$.
Note that the variable changes for the cases $n=2$ and $n=3$ can be understood as the canonical changes, $x\to x^2$ and $x\to \left(x^2+1 \over 2x\right)^2$ followed by an arbitrary rational change $x\to p/q$ (which can be discarded, see below).

Now we argue that, apart from the above reasons, the rational transformations can not help in finding $\epsilon$-form. I.e., if all eigenvalues of the matrix residues are already of the form $k+\alpha \epsilon\,,\ k\in\mathbb{Z}$, the $\epsilon$-form either can be reached by a rational in $x$ transformation $T$ or can not be reached by any transformation of a wider class of rational transformations in $y$. 

First we note that if the system is normalized Fuchsian at some point $x=x_0$ with the normalization given by the condition $\lambda =\alpha \epsilon$, the new system obtained by change of variable \eqref{eq:varchange} is also normalized Fuchsian at any preimage of $x_0$.

Suppose, that $U$ has the decomposition \eqref{eq:decomposition1} with $D\neq 0$. Then without loss of generality we may assume that
\begin{equation}\label{eq:decomposition2}
U =x^D\,.
\end{equation}
Indeed, the matrices $M_T$  and $\tilde{M}_S$ have the same properties as $M$  and $\tilde{M}$. So, we may replace $M\to M_T$  and $\tilde{M}\to\tilde{M}_S$ everywhere in our considerations. Then $U$ has the form presented in \eqref{eq:decomposition2}. Each nonzero eigenvalue of $D$ corresponds to the degree of a holomorphic line bundle. Then, it is easy to see that upon the rational change of variables $x=p(y)/q(y)$, the degree of the line bundle is multiplied by the degree of the mapping $k=\max(\deg p(y),\deg q(y))$. Therefore, the matrix $U$ corresponds to the holomorphic vector bundle on the Riemann sphere of $y$ labeled by the spectrum of $k D$. Since, by assumption, $k D\neq 0$, normalization can not be performed by means of the transformations rational in $y$.

\section{Examples}

Let us give two examples of the application of the presented algorithm. Both examples concern reduction to $\epsilon$-form for the multiloop calculations and we fix the normalization by the requirement that all eigenvalues of all matrix residues are proportional to $\epsilon$.

\paragraph{Example 1.}
First, consider the matrix 
\begin{equation}
	M=\left(
	\begin{array}{ccc}
	\frac{2 \left(2 x^2+5\right) \epsilon }{(x-1) x (x+1)} & \frac{\epsilon }{(x-1) (x+1)} & \frac{(x+3) \epsilon }{(x-1)
		(x+1)} \\
	-\frac{12 (x+9) \epsilon }{(x-1) (x+1)} & -\frac{x (x+7) \epsilon }{(x-1) (x+1)} & -\frac{x^3 \epsilon +12 x^2 \epsilon
		+x^2+35 x \epsilon -1}{(x-1) (x+1)} \\
	\frac{12 \epsilon }{(x-1) (x+1)} & \frac{x \epsilon }{(x-1) (x+1)} & \frac{x (x+5) \epsilon }{(x-1) (x+1)} \\
	\end{array}
	\right)
\end{equation}
This matrix has singularities at $x=0,-1,1,\infty$ and is normalized Fuchsian at $x=0,-1,1$. At $x=\infty$ it is not Fuchsian. Now we reduce this matrix by applying a few balances between $0$ and $\infty$ trying to obtain the normalized Fuchsian form at $x=\infty$ but not paying attention to the behavior at $x=0$. In particular, we may  apply the transformation 
\begin{equation}
U=\left(
\begin{array}{ccc}
1 & 0 & 0 \\
0 & x & 0 \\
0 & -\frac{x-1}{x} & \frac{1}{x} \\
\end{array}
\right)\,
\end{equation}
and obtain 
\begin{equation}
M_U=\left(
\begin{array}{ccc}
\frac{2 \left(2 x^2+5\right) \epsilon }{(x-1) x (x+1)} & -\frac{(2 x-3) \epsilon }{(x-1) x (x+1)} & \frac{(x+3)
	\epsilon }{(x-1) x (x+1)} \\
-\frac{12 (x+9) \epsilon }{(x-1) x (x+1)} & \frac{4 x^3 \epsilon +23 x^2 \epsilon -x^2-35 x \epsilon +1}{(x-1) x^2
	(x+1)} & -\frac{x^3 \epsilon +12 x^2 \epsilon +x^2+35 x \epsilon -1}{(x-1) x^2 (x+1)} \\
-\frac{12 (8 x-9) \epsilon }{(x-1) x (x+1)} & \frac{24 x^3 \epsilon -58 x^2 \epsilon +x^2+35 x \epsilon -1}{(x-1) x^2
	(x+1)} & -\frac{6 x^3 \epsilon +23 x^2 \epsilon -x^2-35 x \epsilon +1}{(x-1) x^2 (x+1)} \\
\end{array}
\right)
\end{equation}
This matrix is normalized Fuchsian at $x=-1,1,\infty$ but not at $x=0$. Now we apply the decomposition algorithm to the matrix $U$ and obtain:
\begin{equation}
	U=T(x)S^{-1}(x^{-1})=\left(
	\begin{array}{ccc}
	1 & 0 & 0 \\
	0 & 1 & x \\
	0 & 0 & -1 \\
	\end{array}
	\right) \times\left(
	\begin{array}{ccc}
	1 & 0 & 0 \\
	0 & 1 & 1 \\
	0 & 1-x^{-1} & -x^{-1} \\
	\end{array}
	\right)
\end{equation}
Applying the transformation $T$, we find
\begin{equation}
M_T=\left(
\begin{array}{ccc}
\frac{2 \left(2 x^2+5\right) \epsilon }{(x-1) x (x+1)} & \frac{\epsilon }{(x-1) (x+1)} & -\frac{3 \epsilon }{(x-1)
	(x+1)} \\
-\frac{108 \epsilon }{(x-1) (x+1)} & -\frac{7 x \epsilon }{(x-1) (x+1)} & \frac{35 x \epsilon }{(x-1) (x+1)} \\
-\frac{12 \epsilon }{(x-1) (x+1)} & -\frac{x \epsilon }{(x-1) (x+1)} & \frac{5 x \epsilon }{(x-1) (x+1)} \\
\end{array}
\right)\,
\end{equation}

\paragraph{Example 2.}
Let us now consider the matrix
\begin{equation}
M=\left(
\begin{array}{cc}
-\frac{(x+1) \epsilon }{(x-1) x} & -\epsilon  x+x-3 \epsilon -3 \\
-\frac{1}{(x-9) (x-1) x} & -\frac{2 \epsilon }{x-9} \\
\end{array}
\right)
\end{equation}
This matrix has singularities at $x=0,1,9,\infty$ and is normalized Fuchsian at $x=0,1,9$. At $x=\infty$ it is not Fuchsian. We reduce this matrix by the transformation 
\begin{equation}
U=\left(
\begin{array}{cc}
x & x-1 \\
0 & x^{-1} \\
\end{array}
\right)
\end{equation}
and obtain 
\begin{equation}
M_U=\left(
\begin{array}{cc}
-\frac{\epsilon  x^2-8 \epsilon  x-9 x-9 \epsilon +9}{(x-9) (x-1) x} & \frac{6 (x+3) (2 \epsilon +1)}{(x-9) x^2} \\
-\frac{x}{(x-9) (x-1)} & -\frac{2 x \epsilon +9}{(x-9) x} \\
\end{array}
\right)
\end{equation}
This matrix is normalized Fuchsian at $x=1,9,\infty$ but not at $x=0$. Now we apply the decomposition algorithm to the matrix $U$. As it is required, both $U$ and $U^{-1}$ are Laurent polynomials of $x$ and $\det U=1$ is independent of $x$. However, there is only one polynomial vector $v(x)={1\choose 0}$ which satisfies Eq. \eqref{eq:PPeq}. Therefore, we can not find the required decomposition. This means that the global normalized Fuchsian form does not exist (and the more so for the $\epsilon$-form).


\section{Conclusion}

In the present paper we have formulated the criterion of (ir)reducibility of the system of differential equations to global normalized Fuchsian form on the Riemann sphere. This criterion is constructive in the sense that for reducible system it gives the required transformation matrix.

The question of reducibility is known to be related to that of triviality of holomorphic vector bundle on the Riemann sphere \cite{Deligne}. In application to $\epsilon$-form in multiloop calculations this perspective allows one to establish (ir)reducibility to $\epsilon$-form by the transformations from a wider class of transformations rational in arbitrary new variable $y$ connected with the old variable $x$ by \eqref{eq:varchange} thus sparing the necessity to consider passing to new variable for the systems which can be reduced to $\epsilon$-form locally but not globally.

We have considered the question of choosing the suitable new variable $y$, see Eq. \eqref{eq:varchange}, for the systems which can not be locally reduced to $\epsilon$-form in some points. As a rule, this happens because some of the eigenvalues of the matrix residues in these points have the form $k+1/2+\epsilon$. We have shown that in case of two and three such points there is a canonical variable change while for $n>3$ points there is no suitable variable at all. In case of only one point with matrix residue having 'half-integer' eigenvalue there is a wider class of variable changes which one should try: $x = p^2(y)/q(y)$, where $p$ and $q$ are some polynomials, $\deg q\leqslant \deg p$. It would be interesting to further restrict the analysis in this case.

Two possible directions of further investigation are obvious. First, one might stick to the rational transformations, possibly augmented by the rational variable change, and ask what generalizations of the $\epsilon$-form are required for the systems which are not reducible. One of the possible generalizations would be to try to represent the matrix in the right-hand side as
\begin{equation}
M(\epsilon,x)=\epsilon A(x)+B(x)
\end{equation} 
with $B\neq 0$. In particular, the equations for massive sunrise diagrams can usually be reduced to $\epsilon$-form near $d=3$. It means that we can secure $M(\epsilon,x)=(\epsilon -1/2) A(x)$.

Another direction would be to examine how one can minimally extend the class of transformations $T$ which are required to reduce the systems to $\epsilon$-form.

\acknowledgments

The work of R. Lee has been supported by the grant of the ``Basis'' foundation for theoretical physics.

\bibliographystyle{JHEP}

\providecommand{\href}[2]{#2}\begingroup\raggedright\endgroup
\end{document}